\documentclass[a4paper,11pt]{article}

\usepackage[top=1in, bottom=1in, left=1in, right=1in]{geometry}
\usepackage[pdfborder={0 0 0}]{hyperref}
\usepackage[ruled,vlined,linesnumbered]{algorithm2e}
\usepackage{amsmath, amssymb, amsthm}
\usepackage{graphicx}
\usepackage{subcaption}
\usepackage{thm-restate}
\usepackage{enumitem}

\newtheorem{theorem}{Theorem}[section]
\newtheorem{lemma}[theorem]{Lemma}
\newtheorem{claim}[theorem]{Claim}

\theoremstyle{definition}
\newtheorem{definition}[theorem]{Definition}
\theoremstyle{plain}

\newcommand{\bbR}{\mathbb{R}}
\newcommand{\bbZ}{\mathbb{Z}}
\newcommand{\Exp}{\mathop{\mathbb{E}}}
\newcommand{\eps}{\varepsilon}
\newcommand{\IR}{\mathbb{R}}
\newcommand{\OPT}{\mathtt{opt}}

\newcommand{\cd}{\texttt{cd}}

\newcommand{\wdep}{\texttt{wd}}
\newcommand{\cD}{\mathcal{D}}
\newcommand{\cU}{\mathcal{U}}
\newcommand{\cB}{\mathcal{B}}
\newcommand{\cG}{\mathcal{G}}
\newcommand{\disks}{\mathcal{D}}
\newcommand{\arcs}{\mathcal{A}}
\newcommand{\Oe}{O_{\eps}}
\newcommand{\Ot}{\widetilde{O}}

\title{A Bouquet of Results on Maximum Range Sum: General Techniques and Hardness Reductions}

\author{
  Rachana Gusain\thanks{Rachana Gusain is supported by the Prime Minister's Research Fellowship (PMRF). Email: \texttt{rachanag@iisc.ac.in}} \and
  Saladi Rahul\thanks{Saladi Rahul would like to thank ``Walmart Center for Tech Excellence'' and ``Anusandhan National Research Foundation (ANRF) CRG/2023/005776''. Email: \texttt{saladi@iisc.ac.in}} \and
  Aditya Subramanian\thanks{Email: \texttt{adityasubram@iisc.ac.in}}
}
\date{Indian Institute of Science, Bengaluru}

\begin{document}
\maketitle
\begin{abstract}
In this work we revisit the {\em maximum range sum} (MaxRS) problem which is well studied by the database and the computational geometry communities. The input is a set $P$ of $n$ weighted points in $\IR^d$ and a geometric range $Q$ (typically either an axis-aligned $d$-box or a $d$-ball). The goal is to design a fast algorithm to place $Q$ in $\IR^d$ so that the total weight of the points of $P$ inside $Q$ is maximized. We consider three natural variations of the MaxRS problem:

\begin{enumerate}
    \item In the {\em dynamic} MaxRS problem, points are inserted and deleted, and the goal is to
        efficiently update the placement of a $d$-ball. In $\IR^d$ we present a randomized
        $(\frac{1}{2}-\eps)$-approximation algorithm with update time $\Oe(\log n)$.  The
        approximation factor holds with high probability. To the best of our knowledge, this problem
        was not studied before in the literature.
    
    \item In the {\em batched} MaxRS problem in $\IR^1$, along with the points in $P$ we are given
        $m$ intervals of different lengths. The goal is to solve the MaxRS problem for each
        interval. We establish a conditional lower bound of $\Omega(mn)$ time for this problem,
        assuming the hardness of $(\min,+)$-convolution problem. Interestingly, this implies that
        the trivial upper bound of $O(mn\log n)$ for batched MaxRS in $\IR^2$ is almost-tight. A
        similar lower bound is established for a related problem of {\em batched smallest
        $k$-enclosing interval}.
    
    \item In the {\em colored} MaxRS problem in $\IR^d$, each point in $P$ is assigned a color from
        $\{1,2,\ldots,m\}$ and the goal is to find the placement of a $d$-ball $Q$ that maximizes
        the number of \emph{uniquely colored} points in $P\cap Q$. Prior work on this problem was
        limited to axis-aligned rectangle $Q$ in $\IR^2$. We obtain two new results for $d$-balls.
        The first result is a randomized $(\frac{1}{2}-\eps)$-approximation algorithm with running
        time $\Oe(n\log n)$. Interestingly, the exponential dependence of $\log n$ on $d$ is avoided
        in the running time. The second result improves upon the first result in $\IR^2$ by
        providing a $(1-\eps)$-approximation algorithm with expected running time $\Oe(n\log n)$.
        The approximation factor holds with high probability for both results.

\end{enumerate}

Our algorithms are obtained via two general techniques which we believe will be useful for solving other variants of MaxRS. The first technique provides a $(\frac{1}{2}-\eps)$-approximation guarantee. The analysis relies on a volume argument involving $d$-balls and a randomized game. The second technique provides a $(1-\eps)$-approximation guarantee and works in two phases. In the first phase, we design an exact output-sensitive algorithm, and in the second phase, we speed up the exact algorithm by random sampling on colors.
\end{abstract}

\newpage
\section{Introduction}
{\em Maximum range sum} (MaxRS) is a popular geometric placement problem that has been widely studied by the spatial database community~\cite{choi2012scalable, choi2014maximizing, tao2013approximate, AmagataH16, amagata2017general, MostafizMHAT17, nakayama2017probabilistic, liu2019probabilistic, liu2024sweepline, zhang2022maximizing, hussainEBDT17contMaxRS, tahmasbi2022dynamic} and the computational geometry community~\cite{imai1983finding, nandy1995unified, chazelle1986circle, bdp97, agarwal2002translating, chan2008slightly, de2009covering, jin2018near, zhang2025approximation}. The input is a set $P$ of $n$ weighted points in 
$\IR^d$ and a geometric range $Q$ (typically either an 
axis-aligned $d$-box or a $d$-ball). The goal is to design
a fast algorithm to place $Q$ in $\IR^d$ so that the 
total weight of the points of $P$ inside $Q$ is maximized.
In the unweighted setting, the goal is to place $Q$ 
to cover the maximum number of points. See Figure~\ref{fig:maxrs-ex1}.
In our work, the dimension $d$ is small (say $5$ or $8$)
and treated as a constant.

\begin{figure}[!ht]
\centering
\begin{subfigure}{0.4\textwidth}
\centering
\includegraphics[width=0.7\textwidth]{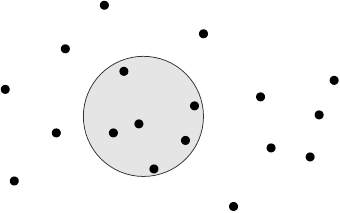}
\caption{MaxRS with unit weights. $\OPT=6$.}
\label{fig:maxrs-ex1}
\end{subfigure}
\hspace{0.05\textwidth}
\begin{subfigure}{0.4\textwidth}
\centering
\includegraphics[width=0.7\textwidth]{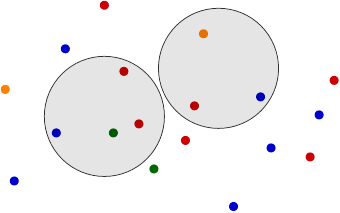}
\caption{Colored MaxRS with $\OPT=3$.}
\label{fig:maxrs-ex2}
\end{subfigure}
\caption{MaxRS with $Q$ being disks.}
\label{fig:maxrs-example}
\end{figure}

MaxRS shows up naturally in many spatial data analysis tasks, 
which involve detecting hotspots or regions of high activity or concentration.
For example, during COVID times there were applications built by 
government agencies which had real-time data of the 
geographic location of the people currently 
infected
and would display ``hotspots'' which had maximum 
patients. MaxRS can help 
political parties and companies make policy and 
infrastructure decisions. For example,
Walmart has spatial data which contains the geographic
location of its costumers. A MaxRS query
can reveal the region with maximum number of costumers
and hence, the decision of opening a new outlet can be made.

 The MaxRS problem has primarily been studied for two classes of geometric ranges: axis-aligned $d$-dimensional hyper-rectangles
 ($d$-boxes for brevity) and $d$-dimensional euclidean balls 
 ($d$-balls for brevity). The problem can trivially be solved 
 in polynomial time (in terms of $n$) 
 for $d$-boxes and $d$-balls, where $d$ is a constant. 
 For example, for $2$-box (a.k.a. rectangle in $\IR^2$)
 the optimal placement of $Q$ is defined by at most four points, 
 one point on each edge of the rectangle. The work on MaxRS 
 was first initiated by~\cite{imai1983finding} who presented
 an exact optimal $O(n\log n)$ time algorithm for $2$-box 
 (and later by~\cite{nandy1995unified}). The extension 
 to $d$-box ($d\geq 3$) was done in \cite{chan2008slightly} with 
 running time $\widetilde{O}({n^{d/2}})$. For $2$-ball 
 (a.k.a. disk in $\IR^2$), Chazelle and Lee presented an 
 exact $O(n^2)$ time algorithm~\cite{chazelle1986circle}, 
 which was later shown to be optimal~\cite{aronov2008approximating} (assuming hardness of 3SUM problem). Choi, Chung and Tao~\cite{choi2014maximizing}
 perform an extensive study of the MaxRS problem in the 
 I/O-model.
 \cite{esz94,bdp97,ChanH21,zhang2025approximation, FigueiredoF09almost-unitball, EcksteinHLNS02maxbox-data-analysis} are some other relevant
 papers on this problem.
 
Parallel efforts have been 
made to design $(1-\eps)$-approximate solutions with near-linear 
running time in terms of $n$. 
Given an $\eps$ satisfying $0 < \eps < 1$, the goal in the $(1-\eps)$-approximate MaxRS problem is to place $Q$ so that the total weight of points in $P\cap Q$ is at least $(1-\eps)\cdot\OPT$, where $\OPT$ is the total 
weight in the optimal placement of $Q$. 
These include random 
sampling based algorithms~\cite{agarwal2002translating, aronov2008approximating,AgarwalS21,tao2013approximate} 
and deterministic algorithms based on 
$\eps$-approximations~\cite{de2009covering, jin2018near}.

In this work we consider three natural variants of the 
MaxRS problem: dynamic MaxRS, batched MaxRS and colored MaxRS.

\subsection{Dynamic MaxRS}
In the {\em dynamic} MaxRS problem, a new 
point is inserted into $P$ or an existing point from $P$ is deleted
in each round, and the goal is to design an algorithm which can 
quickly update the optimal placement of $Q$. The naive approach
would be to re-compute the solution from scratch which requires
$\Omega(n)$ time, whereas the goal is to re-compute the solution
in sub-linear time. Note that $Q$ does not change in this setting.
Going back to the COVID example, location of 
recovered patients are deleted
and location of 
newly infected patients are inserted, and the authorities
need to be reported about the updated hotspot in real-time.
To the best of our knowledge, dynamic MaxRS has not been 
explicitly studied before. We obtain the following result.

\begin{theorem}\label{thm:dyn-d-ball}
For dynamic MaxRS with $Q$ being a $d$-ball, 
there is a randomized $(\frac{1}{2}-\eps)$-approximation
algorithm 
with amortized update time $O(\eps^{-2d-2}\log n)$, 
for $0 < \eps < 1/2$. 
The approximation factor holds with high probability, 
i.e., $1-n^{-\Omega(1)}$. Here $n$ is the number of 
points of $P$ in a given round.
\end{theorem}

As a consequence of the above theorem, we obtain a new 
result for static MaxRS (where the set $P$ is known 
upfront) with $Q$ being a $d$-ball.

\begin{theorem}\label{thm:stat-d-ball}
For static MaxRS with $Q$ being a $d$-ball, 
there is a randomized $(\frac{1}{2}-\eps)$-approximation
algorithm 
with running time $O(\eps^{-2d-2}n\log n)$, 
for $0 < \eps < 1/2$. 
The approximation factor holds with high probability.
\end{theorem}

Previous constructions
can be extended to obtain a $(1-\eps)$-approximation algorithm
for static MaxRS with $Q$ being a $d$-ball, but they have a running
time of $O_{\eps}(n\log^{\Theta(d)}n)$.
Interestingly, in Theorem~\ref{thm:stat-d-ball} 
the exponential dependence of
$\log n$ on $d$ is avoided in the running time (at the expense
of approximation factor).

\subsection{Batched MaxRS}

In the {\em batched} MaxRS problem, along with the point set $P$, we are given a set
$\{Q_1,Q_2,\dots,Q_m\}$ of $m$ geometric ranges, 
and the goal is to find a placement for each range $Q_i$
which maximizes the weight of points of $P$ covered by it. 
Our original motivation for studying batched MaxRS
was $Q_i$'s being axis-aligned rectangles in $\IR^2$. 
By running the $O(n\log n)$ time algorithm~\cite{nandy1995unified,imai1983finding}
for MaxRS for each $Q_i$, the batched MaxRS problem 
can be solved in $O(mn\log n)$ time.
The following natural question arises:
\begin{center}
{\em Does there exist an $o(mn)$ time algorithm for batched
MaxRS with axis-aligned rectangles in $\IR^2$?}
\end{center}    

The punchline is that even in $\IR^1$ this looks 
unlikely. Specifically, we establish an $\Omega(mn)$ time 
conditional lower bound for batched MaxRS in $\IR^1$
assuming hardness
of the popular {\em $(\min,+)$-convolution} problem.

\begin{sloppypar}
\paragraph{The $(\min,+)$-convolution problem.}
Given two sequences $A = (A_0, A_1, \dots, A_{n-1})$ and $B = (B_0, B_1, \dots, B_{n-1})$, their
$(\min,+)$-convolution produces a new sequence \( C \), where each element is defined as:
\[ C_k = \min_{i + j = k} (A_i + B_j), \]
for $k\in\{0,\dots,n-1\}$.
\end{sloppypar}

This problem admits a trivial quadratic time algorithm, but has been conjectured to be quadratic hard,
i.e., there exists no $O(n^{2-\epsilon})$ algorithm for it, for any constant
$\epsilon>0$~\cite{CyganMWW19}. Till date the conditional 
hardness of sevaral problems has been established via
the $(\min,+)$-convolution problem~\cite{LaberRC14,KunnemannPS17,BackursIS17,CyganMWW19,ChanH21,JansenR23}.

\begin{theorem}\label{thm:batched-maxrs}
Assuming that $(\min,+)$-convolution on sequences of length $n$ requires $\Omega(n^2)$ time, the batched MaxRS problem on $n$ points and $m$ query interval lengths
requires $\Omega(mn)$ time.
\end{theorem}

A similar lower bound is established
for a related problem of {\em batched smallest $k$-enclosing
interval}. In the smallest $k$-enclosing interval (SEI) problem,
the input is a set $P$ of $n$ points on the real-line and 
and an integer $k$ in the range $[1,n]$. The goal is to 
find the smallest interval which contains $k$ points of $P$.
In the batched smallest $k$-enclosing interval problem, 
the goal is to solve the SEI problem
for {\em all} integer values of $k$ from $1$ till $n$. 
Once again the problem can be solved trivially 
in $O(n^2)$ time. We establish a conditional
lower bound  of $\Omega(n^2)$ for this problem as well.

\begin{theorem}\label{thm:bsei}
Assuming that $(\min,+)$-convolution on sequences of length $n$ requires $\Omega(n^2)$ time, the
batched smallest k-enclosing interval problem on $n$ points requires $\Omega(n^2)$ time.
\end{theorem}

\subsection{Colored MaxRS}
In the {\em colored} MaxRS problem in $\IR^d$, each point
in $P$ is assigned a color from $\{1,2,\ldots,m\}$
and the goal is to find a placement of $Q$ that
maximizes the number of \emph{uniquely colored} points 
in $P\cap Q$. Note that the points no longer have a weight 
associated. See Figure~\ref{fig:maxrs-ex2} for an 
example.
Colored MaxRS was recently studied in the context of 
trajectory data~\cite{zhang2022maximizing}. They present an exact $O(n\log{n})$ time 
algorithm for $Q$ being an axis-aligned rectangle.

In the context of wildlife 
conservation, consider 
$m$ trajectories each corresponding to say an 
endangered animal. Points are sampled from each trajectory
and assigned a unique color. The objective is to optimally 
position a tracking device to monitor the maximum number of 
animals. See \cite{zhang2022maximizing} for more examples.
Colored MaxRS is also useful for finding neighborhood with 
maximum number of distinct facilities (such as restaurants,
schools, hospitals, parks and fire stations). 

In this work we primarily focus on colored MaxRS with 
$Q$ being $d$-balls. 
Recall that even (uncolored) MaxRS for $Q$ being a $2$-ball
has a conditional lower bound of $\Omega(n^2)$~\cite{aronov2008approximating} 
(and in fact 
we conjecture a conditional lower bound of $\Omega(n^d)$
for colored MaxRS for $Q$ being a $d$-ball).
Therefore, we focus on designing approximation algorithms
with near-linear running time. The goal in the
$\alpha$-approximate colored MaxRS problem is to place
$Q$ so that the number of uniquely colored points in 
$P \cap Q$ is at least $\alpha \cdot \OPT$, where 
$\OPT$ is the number of uniquely colored points in the 
optimal placement of $Q$.
The first result is the following.

\begin{restatable}{theorem}{colhalfapprox}\label{thm:col-balls-half-approx}
For colored MaxRS with $Q$ being a $d$-ball, there is 
a randomized $(\frac{1}{2}-\eps)$-approximation algorithm 
with running time $O(\eps^{-2d-2}n\log n)$, for 
$0 < \eps < 1/2$. The approximation factor
holds with high probability.
\end{restatable}

The second result improves upon the first result in $\IR^2$
in terms of approximation factor.

\begin{theorem}\label{thm:col-balls-1minus-eps-approx}
For colored MaxRS with $Q$ being a $2$-ball, 
there is 
a randomized $(1-\eps)$-approximation algorithm 
with expected running time $O(\eps^{-2}n\log n)$, for 
$0 < \eps < 1$. The approximation factor
holds with high probability.
\end{theorem}

\subsection{The dual setting: computing the maximum depth}
Till now we have described the MaxRS problem in the {\em primal}
setting. By appropriate scaling, we can assume that 
the $d$-ball $Q$ has unit radius. We now state the 
MaxRS problem in the {\em dual} setting. Replace 
each point $p\in P$ (with weight $w_p$) with 
a unit radius $d$-ball $B(p)$ (with weight $w_p$) 
centered at $p$. Let $\cB \leftarrow \bigcup_{p \in P}B(p)$ be the
collection of these $n$ weighted unit balls. 
Then the MaxRS problem for $Q$ being a 
$d$-ball is equivalent to the problem of finding the 
point in $\IR^d$ with the maximum {\em weighted depth} 
w.r.t. $\cB$.
Similarly, colored MaxRS in the dual setting will be a 
set $\cB$ of $n$ colored unit $d$-balls and the goal is 
to report the point in $\IR^d$ with the maximum
{\em colored depth}. The colored-depth of a point $p$, 
say $\cd_{\cB}(p)$, is the number of uniquely colored balls 
in $\cB$ that contain $p$. 

\subsection{An overview of our techniques}

MaxRS is amenable to sampling, i.e., in the dual setting a 
point in $\IR^d$ with large (resp., small) depth w.r.t. $\cB$ 
is expected to have a large (resp., small) depth wr.r.t. a 
uniform random sample in $\cB$. Prior works obtained a
$(1-\eps)$-approximation by exploiting this property.
The strategy works in two stages: 
ignoring the dependency on $\eps$ and assuming the value 
of $\OPT$ is known, first compute a sample of 
size roughly $O(\log n)$, and then run the exact algorithm on 
the sample. Unfortunately, for a $d$-ball in the
dynamic setting, this approach leads to a 
$(1-\eps)$-approximation algorithm with an 
update time that has exponential dependence on $d$ (i.e., 
$\log^{O(d)}n$), since exact MaxRS with a $d$-ball on $n$ points requires $O(n^d)$ time. 
In contrast, our first technique achieves an update time that 
does {\em not} have an exponential dependence on $d$ with respect to $\log n$.

Our first general technique which obtains a 
$(\frac{1}{2}-\eps)$-approximation follows a different approach.
Instead of sampling the input objects, the algorithm 
samples a bunch of points in $\IR^d$, maintains their 
(weighted or colored) 
depth w.r.t. $\cB$ and reports the maximum among them.
The algorithm is simple but the analysis entails a 
randomized game involving a set of fixed objects in $\IR^d$ 
and a random sample of points, and a ``volume'' argument
involving $d$-balls. This technique is used to 
obtain Theorem~\ref{thm:dyn-d-ball}, \ref{thm:stat-d-ball}, 
and \ref{thm:col-balls-half-approx}.

Our second technique provides a $(1-\eps)$-approximation 
guarantee for colored MaxRS problem (previous 
works gave a $(1-\eps)$-approximation 
guarantee for un-colored MaxRS). The algorithm works 
in two phases. In the first phase, we design an
exact output-sensitive algorithm, and in the second phase,
we speed up the exact algorithm by random sampling on 
colors. As an application, we consider colored MaxRS 
with $Q$ being a disk in $\IR^2$.
There is a straightforward $O(n^2\log n)$ time algorithm
to solve the problem. 
We first design an algorithm whose running time is proportional
to the value of $\OPT$. 
Specifically, the running 
time is $O(n\log n+n\cdot \OPT)$ (Theorem~\ref{thm:second-algo}). 
Using this output-sensitive algorithm,
we then design the $(1-\eps)$-approximation algorithm 
with $O_{\eps}(n\log n)$ running time (Theorem~\ref{thm:col-balls-1minus-eps-approx}). 

\subsection{Other related work}
The study of the MaxRS problem was initiated by the computational geometry community~\cite{imai1983finding, nandy1995unified, chazelle1986circle, bdp97, agarwal2002translating, de2009covering, jin2018near, zhang2025approximation}, and more recently it has attracted significant interest in the database community. In the last decade, several works have appeared that study MaxRS and its variants in different models and application settings. For example, the MaxRS problem has been extensively studied in the I/O-model~\cite{choi2012scalable, choi2014maximizing, tao2013approximate}, as well as in spatial data streams~\cite{AmagataH16, amagata2017general, MostafizMHAT17}. An index-based method has been proposed to support queries with varying parameters efficiently~\cite{zhou2016index-maxrs}. To address inherent uncertainty in location data, the probabilistic MaxRS problem has also been explored~\cite{nakayama2017probabilistic, liu2019probabilistic}, and more recently extended to dynamic (varying locations of objects over time) and trajectory data~\cite{liu2024sweepline, zhang2022maximizing, hussainEBDT17contMaxRS}. A different dynamic variant, where object weights decay over time, has been studied in~\cite{tahmasbi2022dynamic}. \cite{Hussain2018areaRS} studied MaxRS for optimizing area coverage of polygons. Variants of the MaxRS problem have also been considered in road networks~\cite{cao2014retrieving, xiao2011optimal, chen2015optimal, liu2016finding, ChenYL19, zhou2016index}, mobile ad hoc networks~\cite{nakayama2016mobile} and wireless sensor networks~\cite{hussain2015wireless}.

A closely related problem is the \emph{best region search} (BRS) problem introduced by~\cite{feng2016towards}, where the goal is to find an optimal rectangle of a specified size that maximizes a user-defined scoring function over the enclosed points. Extensions such as selecting the top-$k$ best regions have also been studied~\cite{skoutas2018efficient, shahrivari2020parallel}. Other region search problems have also been considered in the literature~\cite{feng2019similarregion, Feng2020bursty, liu2020DARS, chen2020optRS}. Some of the classical problems in computational geometry and databases, such as range searching~\cite{matouvsek1994geometric} and range aggregation~\cite{govindarajan2002crb, sheng2011new}, are related but fundamentally different from the MaxRS problem that we study in this paper (in range searching or range aggregation the placement of $Q$ is fixed and summary of $P\cap Q$ is the query).

\section{Preliminaries}\label{sec:prelim}

For any point $x\in\mathbb{R}^d$, we refer to its coordinates as $x_1,\dots,x_d$.  
Denote by $G_s(c)$ the uniform grid with cell side length $s$ 
and the boundaries are defined by $d$ families of hyperplanes.
For each $1\leq i \leq d$, the $i$-th family of hyperplanes
are 
$\{x\in\IR^d \mid x_i = c_i + k\cdot s, \text{ where } k\in\bbZ\}$ 
(perpendicular to the $i$-th dimension and separated by distance
$s$). 

For any point $p\in \IR^d$, consider the cell in 
$G_s(c)$ containing $p$ and let $p'$ be the center of the cell.
Then $p$ is defined to be {\em $\Delta$-near}
if the euclidean distance between $p$ and $p'$ is 
at most $\Delta$. The following is a simple lemma which states
that if we perform enough ``shifts'' of the grid, then
in at least one of the shifted grids, $p$ will be $\Delta$-near.

\begin{lemma}\label{lem:grid-shift} Consider a collection of grids $\cG = \left\{G_s\left(\frac{\Delta}{\sqrt{d}}z\right) \,\mid\, z\in \{0,1,\dots,\frac{s\sqrt{d}}{\Delta}-1\}^d \right\}$. For any point $p\in \IR^d$, there exists at least one grid in $\cG$ in which $p$ is $\Delta$-near.
\end{lemma}

\section{Technique 1: Sampling points in \texorpdfstring{$\IR^d$}{R3}} \label{sec:halfapprox}

In this section we present a general technique for obtaining fast approximations for the MaxRS problem for $d$-balls. We describe our
algorithms in the dual setting, i.e., we are given an input set of $n$ unit
balls $\cB$, and the goal is to find a point that has the maximum possible
depth. Our technique is largely based on the idea of sampling a small set of
points $S$, and estimating depth as the deepest point in this sample. We show
that it is sufficient to consider a sample of $O_\eps(n\log n)$ points to
obtain a $(\frac{1}{2}-\eps)$-approximation. This small sample size allows us to
avoid the exponential dependence of $\log n$ on $d$ for static algorithms,
and also allows us to get fast dynamic update times, since depth of only a small
number of points have to be modified for each update.

\subsection{Dynamic MaxRS with \texorpdfstring{$d$}{d}-ball}
\subsubsection{Algorithm.} In the dynamic setting, balls can be inserted to, or deleted from $\cB$. The goal of the algorithm
is to maintain a point of (approximate) maximum depth.
We first construct a collection of grids $\mathcal{G}$, by applying Lemma~\ref{lem:grid-shift} with
$s=\frac{2\eps}{\sqrt{d}}$ and $\Delta=\eps^2$.  We say a grid cell is {\em non-empty} if at
least one $d$-ball in $\cB$ intersects the cell.

\paragraph{Sampling step.}
 Pick a grid $G\in\mathcal{G}$.  
 For each non-empty cell $X$, let
$C(X)$ be the circumsphere of $X$. We sample a set $S_X$ of 
some $t=\Theta(\eps^{-2}\log |\cB|)$ points on
$C(X)$. Specifically, each point is sampled independently and uniformly at random from
$C(X)$~\cite{muller1959note, weisstein}. 
Iterate over all non-empty cells $X$ to get 
$S_G \leftarrow \bigcup S_X$. Finally,
iterate over all grids $G\in\mathcal{G}$, to obtain
the set  $S\leftarrow \bigcup S_G$.

\paragraph{Epochs and handling updates.}

Our algorithm proceeds in epochs. At the beginning of the 
$j$-th epoch, let the balls in $\cB$ be denoted by $\cB_j$. The $j$-th 
epoch {\em ends}
when the number of balls in $\cB$ go outside the
range $[|\cB_j|/2,2|\cB_j|]$, and the $(j{+}1)$-th epoch starts. 

At the {\em start} of each
epoch, we do a sampling step on $\cB_j$
 to obtain a set of points $S_j$. We then compute
the weighted depth of each point in $S_j$. 
We will show in the analysis that this re-sampling, and
re-computation of depth takes only $O_\eps(|\cB_j|\log |\cB_j|)$ time.

During the {\em course} of an epoch, we need to handle
insertions and deletions of balls. 
Let $B$ be the unit $d$-ball inserted in the current round. For each cell, say $X$, intersected by
$B$, we have two cases. If $X$ is non-empty then the weighted depth of each point in $S_X \cap B$ is increased by $w_B$, the weight of ball $B$. Otherwise, if $X$ is empty, then sample a set $S_X$ of
$t=\Theta(\eps^{-2}\log |\cB_j|)$ points, and any point in $S_X$ contained in
$B$ will have weighted depth $w_B$. The case of deletion can be handled similarly.

\subsubsection{Analysis}
We will start the analysis 
by analyzing a randomized game (Lemma~\ref{lem:game-2-dim}) involving a fixed collection
of $\OPT$ balls and a random sample of points. 
We believe that this game can have wider applications.
For a closed set $Q\subseteq \mathbb{R}^d$, we will use $\partial{Q}$ to denote its boundary.

\begin{lemma}\label{lem:game-2-dim}
Let $C$ be a $d$-ball, and let $\cB' \subseteq \cB$
be a collection of $k$ weighted $d$-balls, such that
the $(d{-1})$-dimensional 
measure of $\partial{C} \cap B$ relative to the 
$(d{-1})$-dimensional 
measure of $C$ is $\frac{1}{2}
-\Theta(\eps)$, for all $B\in \cB'$\footnote{
In $\IR^2$, $C$ will be a disk and any disk in $\cB'$
will cover at least $\left(\frac{1}{2}-\Theta(\eps)\right)$
fraction of $C$'s circumference. 
We also remark that the lemma holds in a more general setting.}. 
Let $S=\{p_1,p_2,\ldots,p_t\}$ be a collection
of points such that each point in $S$ is independently and
uniformly at random picked from $\partial{C}$,
where $t=c\eps^{-2}\log n$ (for a sufficiently
large constant $c$). 
Then with high probability (i.e., $\geq 1-n^{-\Omega(1)}$) there 
exists a point in $S$ whose weighted depth is at least 
$\left(\frac{1}{2}-\Theta(\eps)\right)W$, where $W=\sum_{B\in 
\cB'} w_B$.
\end{lemma}

\begin{proof}
We will use an indirect argument to establish the lemma.
Consider any ball $B \in \cB'$. 
We will first establish that 
with high probability 
at least $\left(\frac{1}{2}-\Theta(\eps)\right)t$ points 
of $S$ will lie inside $B$.
For $1\leq i \leq t$,
let $X_i$ be the indicator random variable which is 
equal to $1$ if $p_i \in B$. Let 
$Y_B=\sum_{i=1}^t X_i$ be the number of points in $S$ contained in $B$.
Then, $\mathbb{E}[Y_B] \geq
\left(\frac{1}{2}-\Theta(\eps)\right)t$.
Since $X_1,X_2,\ldots, X_t$ are independent random 
variables, via standard application of 
Chernoff bound, with 
high probability (w.h.p.),  
we have $Y_B \geq \left(1-\eps\right) 
\left(\frac{1}{2}-\Theta(\eps)\right)t  \geq 
\left(\frac{1}{2}-\Theta(\eps)\right)t$.
Next, by union bound on $k\leq n$ balls,
$Y_B \geq \left(\frac{1}{2}-\Theta(\eps)\right)t$ holds w.h.p. for all $B\in \cB'$.

Now we will connect the random variables $Y_B$
with the weighted depth of points in $S$.
Let $\wdep(p)$ be the weighted depth of any point 
$p \in S$ w.r.t. $\cB'$. 
Then, observe that 
$\sum_{p\in S}\wdep(p)=\sum_{B\in \cB'}w_BY_B$, which by the previous argument is at least $W\left(\frac{1}{2}-\Theta(\eps)\right)t$ w.h.p.
This finishes the proof.
\end{proof}

Motivated by the previous lemma, we now show that if a unit ball passes {\em close to} the center
of a ball of radius $\eps$, then almost half the ``surface area'' of the smaller ball is covered.

\begin{lemma}\label{lem:cap-area}
Let $B$ be a unit radius $d$-ball 
and $C$ be a $d$-ball of radius $\eps$. 
Suppose $\partial{B}$ passes through a point that 
lies at distance $\eps^2$ from the center 
of $C$. Then 
the $(d{-}1)$-dimensional measure
of $B \cap \partial{C}$ 
relative to the $(d{-}1)$-dimensional measure of 
$C$ is at least $\frac{1}{2}-\Theta(\eps)$.
\end{lemma}

\begin{proof}
Via suitable rotations of the $d$-balls, we can 
assume that the center of $C$ is at the origin and 
the center of $B$ is $(0,0,\ldots,0,1+\eps^2)$. 
See Figure~\ref{fig:cap-area}(a).
 For any point $x = (x_1,\dots,x_d)\in \partial{B}\cap\partial{C}$, 
 $\Vert x\Vert = \eps^2$ and 
 $\Vert x-(0,\dots,0,1+\eps^2) \Vert = 1$. Expanding these
 two quantities, we get the following equations:
    \begin{equation*}
        \sum_{i=1}^{d}{x_i^2} = \eps^2 \quad\text{and}\quad
        \sum_{i=1}^{d-1}{x_i^2} + (x_d - 1 - \eps^2)^2 = 1,
    \end{equation*}
which solves to $x_d = \dfrac{3\eps^2 + \eps^4}{2+2\eps^2}$.
We denote the expression on the right by $b$. 
It can be verified that $\eps^2 \leq b \leq 2\eps^2$ 
for all $\eps\in (0,1)$.

In $d=2$, suppose $B\cap \partial{C}$ subtends an angle 
$2\theta$ at the center of $C$ (ref Figure~\ref{fig:cap-area}(a)). 
Since $\cos{\theta} = b/r$, the arc length of $B\cap 
\partial{C}$ relative to the circumference of $C$ is
\begin{align*}
       \frac{2\theta r}{2\pi r}=\frac{1}{\pi}\arccos{\bigg(\frac{3\eps + \eps^3}{2+2\eps^2}\bigg)}
       \geq \frac{1}{\pi}\arccos{(2\eps)}
       \geq \frac{1}{2} - 2\eps,
\end{align*}
for all $0 < \eps < 1/2$.

For $d\geq 3$, we use a known result of computing the surface area of a spherical cap as follows.
Let $h$ denote the height of the spherical cap $B\cap \partial{C}$ determined by the hyperplane $x_d = b$. Define $q = 1 - h/r = b/r = \Theta(\eps)$, since $b = \Theta(\eps^2)$. By an estimate of \cite{chudnov1986minimax, wiki:sphericalcap}, the surface area of $B\cap \partial{C}$ relative to the surface area of $C$ is given by
    \begin{equation*}
        \frac{1}{2} - \frac{G_{d-2}(q)}{2\,G_{d-2}(1)},
    \end{equation*}
    where $G_d(x) = \int_{0}^{x}{(1-t^2)^{(d-1)/2}}\,dt$ for all $x\in[0,1]$. The quantity $G_{d-2}(q)$ can be trivially
    upper bounded by $q=\Theta(\eps)$, and the quantity $G_{d-2}(1)$ 
    can be lower bounded by 
    $\int_{0}^{1}{(1-t^2)^{(d-3)/2}}\,dt \geq 
    \int_{0}^{1/2}{(3/4)^{(d-3)/2}}\,dt =\Omega(1)$.
This finishes the proof.
\end{proof}

\begin{figure}[!ht]
    \centering
\begin{subfigure}{0.45\textwidth}
    \centering
    \includegraphics[width=0.7\textwidth, trim={0 0 0 4cm}, clip]{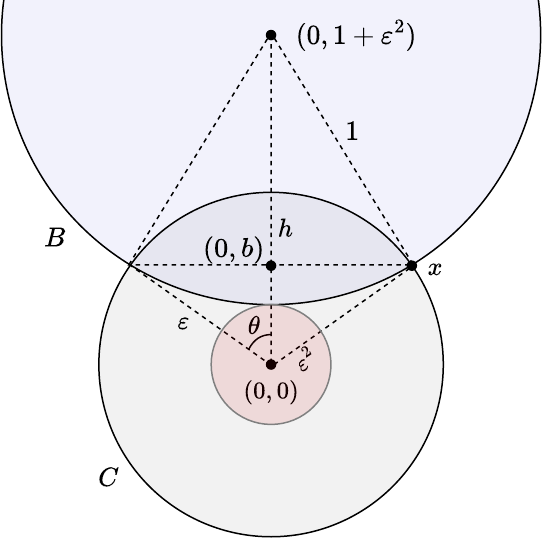}
    \label{fig:cap-area-a}
\end{subfigure}
\begin{subfigure}{0.45\textwidth}
    \centering
    \includegraphics[width=0.7\textwidth, trim={0 0 0 4cm}, clip]{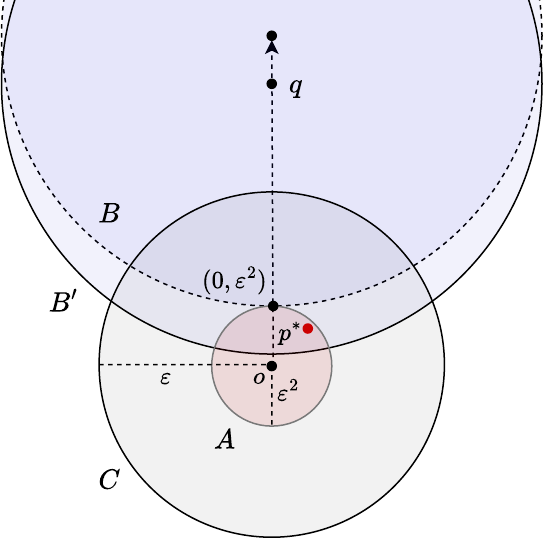}
    \label{fig:cap-area-b}
\end{subfigure}
    \caption{(a) Intersection between a unit disk $B$ centered at $(0,1+\varepsilon^2)$ and a circle $C$ centered at $(0,0)$ of radius $\varepsilon$. 
    (b) Shifting the ball
    $B'$ towards $q$.}
    \label{fig:cap-area}
\end{figure}

Using the above two lemmas, we will now establish the 
approximation factor of our algorithm.

\begin{lemma}
With high probability, the approximation factor of the algorithm is $\frac{1}{2}-\eps$.    
\end{lemma}
\begin{proof}
Consider the grid in $\mathcal{G}$ in which 
the optimal point $p^*$ is $\eps^2$-near 
(guaranteed by Lemma~\ref{lem:grid-shift}).
Let $X^*$ be the cell in the grid containing $p^*$.
W.l.o.g., assume that the center of $X^*$ is at the 
origin $o$. Consider a $d$-ball $A$ of radius $\eps^2$
centered at the origin. Then $p^*$ will lie inside $A$
(see Figure~\ref{fig:cap-area}(b)).
The circumsphere $C := C(X^*)$ of the grid cell has radius $\eps$. Let
$B'\in \cB$ be a ball containing $p^*$ with center $q$.
Translate $B'$ in the direction $\overrightarrow{oq}$ to the ball $B$ until
it is tangential to $A$.
In this case, the spherical cap $B\cap \partial{C} \subseteq B'\cap \partial{C}$,
and therefore, the $(d{-}1)$-dimensional measure of the cap can only decrease. This is the setting considered in Lemma~\ref{lem:cap-area}, which implies that the
minimum measure is at least $\frac{1}{2}-\Theta(\eps)$ relative to the
measure of $C$. We then invoke Lemma~\ref{lem:game-2-dim} on
the ball $C$ and the set of balls $\cB'$ which cover $p^*$,
which shows the existence of a randomly sampled
point in $\partial{C}$ with weighted depth at least $\texttt{opt}(\frac{1}{2} - \Theta(\eps))$ with high probability.
Scaling $\eps$ appropriately gives the desired approximation factor.
\end{proof}

\begin{lemma} \label{lem:updatetime}
The amortized update time of the algorithm is $O(\eps^{-2d-2}\log|\cB|)$.  
\end{lemma}
\begin{proof}
Consider the $j$-th epoch.
First consider the cost of an insertion or deletion 
of a ball. In any fixed grid $G \in \mathcal{G}$, a unit ball intersects at most
    \(
      \left(2/s\right)^d
      = \left(\frac{2\sqrt{d}}{\eps}\right)^d
      = O(\eps^{-d})
    \)
     cells and the collection $\mathcal{G}$ contains
    \(
      \left(\frac{s\sqrt{d}}{\Delta}\right)^d
      = O(\eps^{-d})
    \)
    grids. Therefore, the number of sample points in $S$ updated is \(
      O(\eps^{-2d}) \times O(\eps^{-2}\log |\cB_j|)
      = O(\eps^{-2d-2}\log |\cB_j|)
    \). Hence each update costs
    \(
      O_{\eps}(\log |\cB_j|).
    \)

We now consider the sampling step that will happen
in the $(j{+1})$-th epoch. The number of points sampled
will be $\Oe(|\cB_{j+1}|\log|\cB_{j+1}|)$. Inserting
the balls in $\cB_{j+1}$ one-by-one, and updating the 
weighted depth of $\Oe(\log |\cB_j|)$ sampled points per ball, takes
$\Oe(|\cB_{j+1}|\log|\cB_{j+1}|)$ time.
We will charge the cost of this sampling step to the 
updates which happened in the $j$-th epoch.
Since $|\cB_j|=\Theta(|\cB_{j+1}|)$ and the number of 
updates in the $j$-th epoch are at least $|\cB_j|/2$, 
each update in the $j$-th epoch is charged only 
$\Oe(\log |\cB_j|)$.
\end{proof}


\subsection{Colored MaxRS with \texorpdfstring{$d$}{d}-ball} 
In the dual version of the colored MaxRS problem, we are given a set $\cB$ of $n$ colored unit balls,
and we have to find a point that is covered by maximum number of distinctly colored balls.
Our algorithm is very similar to the earlier described dynamic algorithm. We construct a collection of
grids $\mathcal{G}$, and sample $O_{\eps}(\log n)$ points for each non-empty cell. Only the depth computation step differs slightly.

Let $S$ be the set of sampled points. For each point $p \in S$ maintain a
temporary flag (initialized to $-1$), which records the most recent color
that contributed to its coverage.
To compute the colored depth of each point in $S$, we process the input
balls grouped by their color. Specifically, we order the set $\mathcal{B}$
by color index using any sorting algorithm. Now iterate over the input
balls in this order, updating the colored depth of points corresponding to
cells which intersect the particular ball.

Let $\mathcal{B}_j$ denote the set of balls of color $j \in [m]$. For
each ball in $\mathcal{B}_j$, if a point $p \in S$ is found to lie
inside the ball, and the temporary flag of $p$ is not already equal to
$j$, then we update the flag to $j$, and increment the colored depth
$\cd_{\mathcal{B}}(p)$ by 1. This ensures that the colored depth counts
the number of distinct colors, for which the point is covered by at least one
ball.


\colhalfapprox*
\begin{proof}
    Every unit $d$-ball, intersects $O(\eps^{-2d})$ cells across all grids in $\mathcal{G}$ (as shown earlier in proof of Lemma~\ref{lem:updatetime}),
    and for each cell we maintain a sample of $O(\eps^{-2}\log n)$ points.
    Hence corresponding to each ball, there are $O(\eps^{-2d-2}\log n)$ points that need to be considered.
    Sorting the set of balls takes $O(n\log n)$ time.
    Since corresponding to each ball, we update the colored depth of $O_\eps(\log n)$ points,
    and each update takes $O(1)$ time, we get an overall running time of $O(\eps^{-2d-2}\log n)$.
\end{proof}


\section{Technique 2: Output-sensitivity and color sampling} \label{sec:outsens-colorsample}

In this section we present a $(1-\eps)$-approximation algorithm for 
the colored disk MaxRS problem in $\IR^2$ (Theorem~\ref{thm:col-balls-1minus-eps-approx}).
We have already shown in Section~\ref{sec:prelim} that this problem 
is equivalent to finding a point of maximum
colored depth in a given set of colored unit disks.
So, the input to our problem is a set $\disks$ of $n$
unit disks, each of which is assigned a color from the set $[m]$. 
The goal is to find a point $p$ such that 
$\cd_\disks(p)\ge(1-\eps)\OPT$ (the notation $\cd_{\disks}(p)$ 
refers to the colored depth of a point 
$p$ w.r.t. a set of disks $\disks$).

\subsection{Preliminaries}    
An $x$-monotone arc is a planar curve that is intersected by any vertical line in at most one point. Consider a set 
$\arcs$ of circular $x$-monotone arcs with $k$ 
intersection points. A 
{\em vertical decomposition} of $\arcs$ 
partitions the plane into disjoint cells where each 
cell is a ``pseudo-trapezoid'' with its top and bottom 
segment being a portion from an arc in $\arcs$. The vertical
decomposition is performed in three steps.
First, enclose the arcs inside a large enough bounding box.
Next, shoot a vertical ray upwards and downwards from the 
endpoints of each arc in $\arcs$ and the $k$ intersection 
points. Finally, each ray travels 
until it first hits another arc in $\arcs$
or the top or the bottom of the bounding box. 
See Figure~\ref{fig:transform}(c) for an example.  
The resulting planar subdivision is called the 
{\em trapezoidal map}. We will need the following result.

\begin{lemma}\label{lem:arcs-arrangement}(Theorem 2 in \cite{mulmuley1991fast})
The trapezoidal map a set of $n$ circular $x$-monotone arcs 
with $k$ intersection points can be computed in $O(n\log n + k)$
expected time. The map has $O(n+k)$ vertices, edges and faces.
\end{lemma}

\subsection{The first algorithm}

\paragraph{Transformation to an uncolored problem.}
For a collection $\cD'$ of unit disks, let $\cU(\cD')$ denote the 
{\em union} of the disks $\cD'$, i.e., the set of points in the 
plane which are covered by at least one disk in $\cD'$.
Each $\cU(\cD')$ will be a (possibly disconnected) planar region, which may contain holes.
For each color $c\in[m]$, let $U_c=\cU(\cD_c)$, where $\cD_c$ is the collection
of all disks of color $c$.
The boundary of $U_c$, denoted $\partial{U_c}$, is composed of circular arcs from the corresponding disks of color $c$. 
The depth of a point $p$
in this context is the number of $U_i$'s, for all $1\leq i \leq m$, 
containing $p$. See Figure~\ref{fig:transform}(c). 

We can now transform the original problem into the following
new problem: find a point $p \in \IR^2$ that maximizes its 
(uncolored) depth with respect to the set of regions $\{U_1, \dots, U_m\}$. 

\begin{figure}[h]
\begin{subfigure}{0.31\textwidth}
    \centering
    \includegraphics[scale=1.3,page=1]{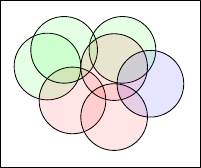}
    \caption{Colored input disks.}
\end{subfigure}
\begin{subfigure}{0.31\textwidth}
    \centering
    \includegraphics[scale=1.3,page=2]{boundary-disks}
    \caption{Merged objects $U_c$.}
\end{subfigure}
\begin{subfigure}{0.31\textwidth}
    \centering
    \includegraphics[scale=1.3,page=3]{boundary-disks}
    \caption{Cells in trapezoidal map.}
\end{subfigure}
 \caption{Reduction from max colored depth in disks to max (uncolored) depth among $U_i$'s.}
 \label{fig:transform}
\end{figure}

\paragraph{Traversal of the trapezoidal map.} 
For each color $c \in [m]$, compute $U_c$ and
determine the circular arcs forming
the boundary $\partial U_c$ \cite{aur88} (see Figure~\ref{fig:transform}(b)).
Each circular arc can be broken into two $x$-monotone arcs
by breaking it at the point with the extreme $x$-coordinate.
Next, apply Lemma~\ref{lem:arcs-arrangement} to compute the trapezoidal map of the circular $x$-monotone arcs 
(see Figure~\ref{fig:transform}(c)).

Now we will traverse the trapezoidal map to compute the 
depth of each cell and identify the cell with the maximum 
depth. Any point inside this cell can be reported as the 
final output.
For each cell $C$ in the trapezoidal map, maintain a variable 
$d_C$
denoting its depth. We will perform a breadth-first search type
traversal of the map starting 
from a cell with depth zero (any cell near the boundary works).
Let $C$ be the current cell visited and 
let $e$ be the common edge when traversing to an adjacent
cell $C'$. If $e$ belongs to one of the vertical rays shot, then
$d_C=d_{C'}$. Otherwise, let $D$ be the disk 
corresponding to edge $e$. If $C\in D$ and $C'\not\in D$, set 
$d_{C'}=d_{C}-1$; else, set $d_{C'}=d_{C}+1$. 

\begin{restatable}{lemma}{firstalgo}\label{lem:first-algo}
    Let $\disks$ be a collection of $n$ colored disks in the plane.
Given regions $U_1,\ldots, U_m$ corresponding to 
each color, define $k$ to be the number of points such that 
the boundaries of any two regions intersect.
Then there is an algorithm to answer colored disk MaxRS 
in $O(n\log n + k)$ expected time. 
\end{restatable}
\begin{proof}
For any color $c\in[m]$, 
the time taken to construct $U_c$ is $O(|\disks_c|\log |\disks_c|)$
via power diagrams~\cite{aur88}. Hence, the time taken to 
construct $U_1,\ldots, U_m$ will be $O(n\log n)$.
The expected time to construct the trapezoidal map is 
$O(n\log n + k)$.
The traversal of the map
can be performed in $O(n+k)$ time by using the standard 
doubly-connected edge list (DCEL) data structure~\cite{dutchBook,cgal2DArrangement}.
The check of $C\in D$ or not
can be performed in constant time by testing if an arbitrary 
point in $C$ lies inside $D$. Overall, the expected time is 
$O(n\log n +k)$.
\end{proof}

\subsection{The second algorithm}
In the second algorithm, we will break the original problem 
into several sub-problems with depth $O(\OPT)$ 
and apply Lemma~\ref{lem:first-algo}
on each sub-problem. This will lead 
to a running time of $O(n\log n + n\cdot\OPT)$.

\paragraph{Algorithm.} We first construct a collection of 
grids $\mathcal{G}$ by applying Lemma~\ref{lem:grid-shift}
with $s=1$ and $\Delta=0.25$.
Fix any one shift of the grid. 
For each non-empty cell $C$ in the grid, 
compute the disks in $\disks$ 
intersecting $C$, say $\disks \cap C$.
Throw away the disks in $\disks \cap C$
which do not contain any corner of cell $C$. 
Based on the remaining disks in $\disks \cap C$,
run the first algorithm (Lemma~\ref{lem:first-algo}) to
identify a point in this grid with 
the maximum depth.  Repeat this 
for all the grid shifts and report the optimal point $p^*$.

\paragraph{Analysis.} We first justify throwing away of 
disks, which allows us to bound the number of colors whose disks
intersect any cell $C$.

\begin{restatable}{lemma}{boundeddepth}\label{lem:bound-depth}
\begin{enumerate}
\item Consider a shift of the grid in which the optimal 
point $p^*$ is $0.25$-near, and lies in some cell $C^*$. A 
unit disk which does not contain any corner of $C^*$ cannot contain $p^*$.
\item For any cell $C$, after throwing away disks, the number of 
uniquely colored disks in $\cD$ intersecting $C$ is at most $4\cdot \OPT$.
\end{enumerate}
\end{restatable}
\begin{proof}
Consider a disk $D$ and a cell $C^*$ as shown in 
Figure~\ref{fig:disk-missing-corners}. Since $|bd| = |cd|$, 
we have $|bd| \leq 1/2$. Now $|ad|=\sqrt{|ab|^2-|bd|^2} \geq 
\sqrt{1-(1/2)^2}\geq \frac{\sqrt{3}}{2}$. Finally, $|ad| + |de| = 1$
implies $|de|\leq 1-\frac{\sqrt{3}}{2}< 0.15$.
\begin{figure}[h]
 \begin{center}
      \includegraphics[scale=0.8]{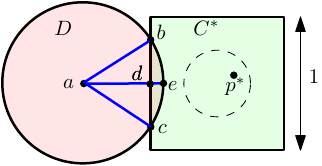}
 \end{center}
 \caption{A disk $D$ not containing any corner of cell $C^*$, cannot contain the optimal point $p^*$.}
 \label{fig:disk-missing-corners}
\end{figure}
This establishes that $D$ cannot contain $p^*$, and completes the proof of (1).

Proof of (2) follows from the fact that the depth of each corner of $C$ is at most $\OPT$. 
\end{proof}

Now we bound the number of intersections among two color classes (Lemma~\ref{lem:two-colors}), and then
use this result to bound the total number of intersections between boundaries of
different color classes (Lemma~\ref{lem:opt-colors}).

\begin{restatable}{lemma}{twocolors}\label{lem:two-colors}
Let $\cD_R$ and $\cD_B$ be a collection of red and blue unit disks,
respectively. 
Let $I(\cD_R,\cD_B)$ be the set of points where there is an 
intersection between an arc in $\partial{U_{R}}$ and an arc in 
$\partial{U_B}$. Then $|I(\cD_R,\cD_B)|=O(|\cD_R| + |\cD_B|)$.
\end{restatable}
\begin{proof}
Consider any point $p \in I(\cD_R, \cD_B)$. 
Let $D_R \in \cD_R$ and $D_B \in \cD_B$ be the 
disks corresponding to the arcs in $\partial{U_R}$ and $\partial{U_B}$, respectively, intersecting at $p$. Since $p \in \partial{U_R} \cap \partial{U_B}$, it implies that $p$ is at a distance
of at least one from the center of any disk in 
$\cD_R \cup \cD_B$. Therefore, by definition of a boundary point, 
we claim that $p \in \partial{(\cD_R \cup \cD_B)}$.

For the sake of contradiction assume that $p$ is not a vertex 
on $\partial{(\cD_R \cup \cD_B)}$. By general position assumption, 
only $D_R$ and $D_B$ contain $p$ and hence, 
either the arc of disk $D_R$ or $D_B$ on $\partial{(\cD_R \cup \cD_B)}$ 
contains $p$. W.l.o.g. assume it is disk $D_R$.
Then $\partial{(\cD_R \cup \cD_B)}$ will have a point 
infinitesimally close to $p$ on either side
 which lie on $D_R$ (as shown in Figure~\ref{fig:boundary-vertex}). 
 Label them $p_a$ and $p_b$.

\begin{figure}[h]
 \begin{center}
      \includegraphics[scale=0.8]{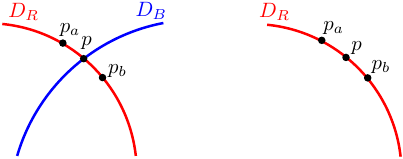}
 \end{center}
 \caption{Proof of Lemma~\ref{lem:two-colors}}
 \label{fig:boundary-vertex}
\end{figure}

Now either $p_a$ or $p_b$ will 
lie at a distance of less than one from the center of $D_B$ 
(in Figure~\ref{fig:boundary-vertex} it is $p_b$). Therefore, 
one among $p_a$ and $p_b$ cannot be on
$\partial{(\cD_R \cup \cD_B)}$. This is a contradiction.
Therefore, $p$ is a vertex on $\partial{(\cD_R \cup \cD_B)}$.
Finally, notice that since $\cD_R \cup \cD_B$ is 
also a collection of disks, the number of vertices 
on $\partial{(\cD_R \cup \cD_B)}$ will be bounded by 
$O(|\cD_R| + |\cD_B|)$~\cite{kedem1986union}.
\end{proof}

\begin{restatable}{lemma}{optcolors}\label{lem:opt-colors}
    Let $\disks$ be a collection of $n$ colored disks in the plane
such that each disk has a color from $t=[4\cdot \OPT]$. 
Given regions $U_1,\ldots, U_t$ corresponding to each color, 
the number of intersection points among the boundaries 
of any two regions is bounded by $O(n\cdot\OPT)$.
\end{restatable}
\begin{proof}
By Lemma~\ref{lem:two-colors} the number of intersection 
points, $|I(\disks_i,\disks_j)|$, among the boundaries
of two regions $\partial{U_i}$ and $\partial{U_j}$
is $O(|\disks_i|+|\disks_j|)$. Summing this quantity over
all the pairs, we get
\[\sum_{i=1}^{t} \sum_{j=1}^{t} O(|\cD_i| + |\cD_j|) = \sum_{i=1}^t O(|\cD_i|\cdot\OPT + n) = O(n\cdot \OPT). \qedhere\]
\end{proof}

\begin{theorem}\label{thm:second-algo}
There is a randomized algorithm for exact colored disk 
MaxRS with $O(n\log{n} + n.\OPT)$ expected
running time. 
\end{theorem}

\begin{proof}
Consider a cell $C$ and let $n_C$ be the number of disks
intersecting $C$. The expected time taken to run the first algorithm
(Lemma~\ref{lem:first-algo}) on these disks
is $O(n_C\log n_C + k_C)$,
where $k_C=O(n_C\cdot \OPT)$ (by Lemma~\ref{lem:opt-colors}).
Note that we can apply Lemma~\ref{lem:opt-colors} since 
by Lemma~\ref{lem:bound-depth}, there are at most $4\cdot \OPT$
colored disks remaining in each cell. Since each disk intersects
at most constant number of cells in the grid, the running
time over all non-empty cells is $O(n\log n + n\cdot \OPT)$.

Now we argue about the correctness. Consider the grid 
in $\mathcal{G}$ where $p^*$ is $0.25$-near.
Let $C$ be the cell containing $p^*$.
By Lemma~\ref{lem:bound-depth}, throwing 
away the disks will not change the colored depth of $p^*$.
One subtle point is that when the first algorithm (Lemma~\ref{lem:first-algo}) is run on the remaining disks 
in $C$, then a point not lying inside $C$ might
get reported. However, by definition of $p^*$ being the optimal
point, the reported point will also have the same colored depth 
as $p^*$.
\end{proof}

\subsection{The final algorithm}
In this section we will design a $(1-\eps)$-approximation
algorithm for colored disk MaxRS via random sampling 
on colors.
Using Theorem~\ref{thm:col-balls-half-approx}
(with $\eps=1/4$),
we first estimate $\OPT$ to get a value 
$\OPT'$ satisfying $\OPT/4 \le \OPT' \le \OPT$. 
If $\OPT' \leq c_1\eps^{-2} \log n$, then run the exact algorithm 
of Theorem~\ref{thm:second-algo} on $\disks$.

Now we consider the interesting case of $\OPT' > c_1\eps^{-2}\log n$.
Sample each color in $[m]$ independently 
with probability $\lambda= \frac{c_1\log n}{\eps^{2}\OPT'}$. 
If a color is sampled, then add all
disks in $\disks$ of that color to the set $R$.
Run the exact algorithm  of Theorem~\ref{thm:second-algo}
on the sampled set $R$ to obtain a point $\hat{x}$ 
which is returned as the solution.

\begin{restatable}{lemma}{samplingruntime} \label{lem:sampling_runtime}
    The algorithm has an expected running time of $O(\eps^{-2}n\log n)$.
\end{restatable}
\begin{proof}
If $\OPT' \leq c_1\eps^{-2}\log n$, then the 
expected running time is $O(n\cdot \OPT)=O(\eps^{-2}n\log n)$.
Now consider the case where $\OPT' > c_1\eps^{-2}\log n$.
For each color $j \in [m]$, let $\disks_j$ be the disks of color $j$. Let $X_j$ be an indicator random
variable such that $X_j=1$ if color $j$ is sampled, and $X_j=0$ otherwise. So we have
$\Exp[X_j] = \lambda$.

The size of $R$ can be written as $|R| = \sum_{j \in [m]} X_j\cdot|\disks_j|$. So, its expected size is
\[\Exp[|R|] = \sum_{j\in[m]} |\disks_j|\cdot \Exp[X_j] = \lambda\cdot\sum_{j\in[m]} |\disks_j| = \frac{c_1n\log n}{\eps^{2}\OPT'}.\]

We run the exact algorithm on the set $R$, giving the desired runtime
$\Exp[|R|\cdot\OPT_R] \le \OPT\cdot\Exp[|R|] = O(\eps^{-2}n\log n)$ in expectation.
\end{proof}

The arrangement of $\disks$ contains 
$O(n^2)$ cells. The colored depth inside each cell
remains the same. A cell is {\em shallow} if
$\cd_{\disks}(\cdot)<(1-\eps)\OPT$, for any 
$0 < \eps < 1$. Pick an arbitrary 
point from each shallow cell and let $P'$ be the 
collection of shallow points. The next lemma 
establishes that with high probability, any shallow 
point in $P'$ will have colored depth $o(\log n)$
w.r.t. $R$. On the other hand, there exists at least one 
point (namely the optimum) that has colored depth $\Omega(\log n)$ w.r.t. $R$. Therefore, running the 
exact algorithm on $R$ will not report a point from $P'$
(with high probability), and hence the output is a $(1-\eps)$-approximation.

\begin{restatable}{lemma}{whpcorrect} \label{lem:whp_correct}
Let $p^*$ be the point with the optimal colored depth 
and let $P'$ be the shallow points w.r.t. $\cD$.  Then, for $c=\frac{(1-\eps/2)\OPT}{\eps^2\OPT'}$,
we have with high probability,
\begin{enumerate}
    \item \emph{$\cd_R(p^*) > c\cdot\log n$}.
    \item for all $p\in P'$, \emph{$\cd_R(p) < c\cdot\log n$}.
\end{enumerate}
\end{restatable}
\begin{proof}
Let $M$ be the set of colors which have a disk covering $p^*$, and for $j\in M$ let  $X_j$ be the random variable indicating whether color $j$ was sampled. So, 
the colored depth of $p^*$ w.r.t. $R$ will be given by $X:=\cd_R(p^*)=\sum_{j\in M}X_j$, and we can compute its expectation as
\[\mu:=\Exp[\cd_R(p^*)]=\Exp\left[\sum_{j\in M}X_j\right]=\sum_{j\in M}\Exp[X_j]=\lambda\cdot\OPT. \]
The concentration of the depth around the expectation can be given by,
\[
 \Pr\left[X\le\left(1-\frac{\eps}{2}\right)\mu\right]
                      \le \exp\left(-\frac{\mu(\eps/2)^2}{2} \right)
                      \le \exp\left( -\frac{c_1\OPT}{8\OPT'}\cdot\log n\right)
                      \le \frac{1}{n^k}
,\]
which holds for $k=\frac{c_1\OPT}{8\OPT'} \geq \frac{c_1}{8} \in \Omega(1)$. So with the high probability of ($1-1/n^k$) we have that $\cd_R(p^*)>(1-(\eps/2))\mu\ge c\cdot\log n$.
This proves (1).
 
Consider some point $p\in P'$. Let $M'$ be the set of colors which have a disk covering $p$,
and for $j\in M'$, let $Y_j$ be
the indicator random variable representing whether color $j$ was sampled. So, 
the colored depth of $p$ w.r.t. $R$ will be given by $Y:=\cd_R(p)=\sum_{j\in M'}Y_j$, and we
can compute its expectation as
\[\mu':=\Exp[\cd_R(p)]=\Exp\left[\sum_{j\in M'}Y_j\right]
 =\sum_{j\in M'}\Exp[Y_j]<\lambda(1-\eps)\OPT. \]

The concentration of the depth around the expectation can be given by,
\[
 \Pr\left[Y\ge\left(1+\frac{\eps}{2}\right)\mu'\right] \le \exp\left(-\frac{\mu'(\eps/2)^2}{3} \right)
        \le \exp\left( -\frac{c_1(1-\eps)\OPT}{12\OPT'}\cdot\log n\right)
        \le \frac{1}{n^{k'}}
,\]
which holds for some $k'=\frac{c_1(1-\eps)\OPT}{12\OPT'} \ge \frac{c_1(1-\eps)}{12} = \Omega(1)$.
So with the high probability of ($1-1/n^{k'}$) we have that
$\cd_R(p)<(1+(\eps/2))\mu' < (1-(\eps/2))\mu \le c\cdot\log n$.

We can now union bound over all $O(n^2)$ points $p\in P'$ to get that
$\Pr[\forall p\in P', \cd_R(p)<c\cdot\log n]\ge 1-\frac{1}{n^{k'-2}}$.
But for a large enough constant $c_1$, we will have that $k'>2$, and hence with high probability,
for all points $p\in P'$ their colored depth w.r.t. $R$ is $\cd_R(p) < c\cdot\log n$. This proves (2). \qedhere
\end{proof}

The above two lemmas establish Theorem~\ref{thm:col-balls-1minus-eps-approx}.

\section{Lower bound for batched MaxRS}
Consider the batched maximum range sum problem (batched MaxRS), where we are given a set of $n$
weighted points $P=\{p_1,p_2,\ldots,p_n\}$,
and a set of $m$ intervals of length $L_1,L_2,\ldots,L_m$. 
For all $1\leq i \leq m$, the goal is to find the placement
of interval of length $L_i$ 
so that the total weight of points of $P$ 
covered is maximized. In this section, we
show a $\Omega(mn)$ lower bound for this problem, conditioned on the conjectured hardness
of the $(min,+)$-convolution problem.
We achieve this via a series of reductions as illustrated in Figure~\ref{fig:reductions}.

\begin{figure}[ht]
    \centering
    \includegraphics[width=0.9\textwidth, page=4]{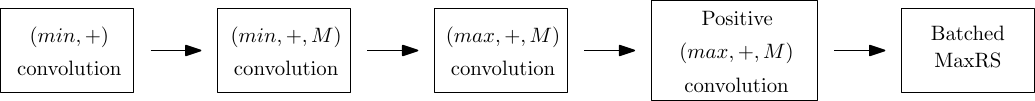}
    \caption{Series of reductions leading to conditional hardness of batched MaxRS.}
    \label{fig:reductions}
\end{figure}

The reduction from $(\min,+)$-convolution to positive $(\max,+,M)$-convolution goes via a series of
straightforward reductions,
the details of which are given in Subsection~\ref{subsec:minpmconv} to \ref{subsec:maxpmconv}.
The technically interesting part of the reduction, which is 
from positive (max,+,M)-convolution to batched MaxRS is given in Subsection~\ref{subsec:batmaxrs}.



\subsection{Reduction from (min,+)-convolution to (min,+,M)-convolution}\label{subsec:minpmconv}
\paragraph{The $(\min,+,M)$-convolution problem.}
\begin{sloppypar}
Given two sequences $D = (D_0, D_1, \dots, D_{n-1})$, $E = (E_0, E_1, \dots, E_{n-1})$, and a set of $m$
distinct indices $M = \{k_0, k_1, \dots, k_{m-1}\}$, where each $k_s \in \{0, 1, \dots, n-1\}$.
The $(\min,+,M)$-convolution problem is to compute the sequence $F_M = (F_{k_s})_{s=0}^{m-1}$ where for
each $k_s \in M$,
\[ F_{k_s} = \min_{\substack{i+j=k_s}} (A_i + B_j). \]
\end{sloppypar}

\paragraph{Reduction:}
Let input sequences $A, B$ be given for an instance of $(\min,+)$-convolution, and we are required to
compute $C_k = \min_{i+j=k} (A_i + B_j)$ for all $k \in K_\text{all}=\{0, 1, \dots, n-1\}$. We will solve
this using an oracle for the $(\min,+,M)$-convolution problem.

We partition the set $K_\text{all}$ into $p = \lceil n/m \rceil$ disjoint subsets $M_1, M_2, \dots,
M_p$. Each subset $M_s$ contains at most $m$ indices from $K_\text{all}$. For instance, we can define
$M_s = \{ i \in K_\text{all} \mid (s-1)m \le i < sm \text{ and } i < n \}$.  For each $s \in \{1,
\dots, p\}$, we make one call to the $(\min,+,M)$-convolution oracle with inputs $A, B,$ and $M_s$. The
union of the results from these $p$ calls yields all the required $C_k$ values for the original
$(\min,+)$-convolution problem.

If we could solve $(\min,+,M)$-convolution in $o(nm)$ time, then we would have solved the $(\min,+)$-convolution
problem in $o(nm)\times n/m = o(n^2)$ time. But the latter problem is conjectured to be
$\Omega(n^2)$ hard, and hence the  $(\min,+,M)$-convolution problem has a conditional
lower bound of $\Omega(nm)$.

\subsection{Reduction from (min,+,M)-convolution to (max,+,M)-convolution}
\paragraph{The $(\max,+,M)$-convolution problem.}
\begin{sloppypar}
Given two sequences $A = (A_0, A_1, \dots, A_{n-1})$, $B = (B_0, B_1, \dots, B_{n-1})$, and a set of $m$
distinct indices $M = \{k_0, k_1, \dots, k_{m-1}\}$, where each $k_s \in \{0, 1, \dots, n-1\}$.
The $(\max,+,M)$-convolution problem is to compute the sequence $C_M = (C_{k_s})_{s=0}^{m-1}$ where for
each $k_s \in M$, \[ C_{k_s} = \max_{\substack{i+j=k_s}} (A_i + B_j). \]
\end{sloppypar}

\paragraph{Reduction:}
Suppose we are given an instance of the $(\min,+,M)$-convolution problem, consisting of sequences $D, E$
(each of length $n$) and a set of $m$ target indices $M$. We aim to compute $F_k = \min_{i+j=k}
(D_i + E_j)$ for each $k \in M$.
We construct two new sequences, $A'$ and $B'$, as follows:
\[ A'_i = -D_i, \quad \text{and} \quad B'_j = -E_j, \quad \text{for } j \in \{0, \dots, n-1\}. \]
Next, we invoke an oracle for the $(\max,+,M)$-convolution problem with the transformed sequences $A'$,
$B'$, and the original set of indices $M$. Let the output from this oracle be denoted by $C'_k$ for
each $k \in M$.
According to the definition of $(\max,+,M)$-convolution, for each $k \in M$:
\[ C'_k = \max_{i+j=k} (A'_i + B'_j) = \max_{i+j=k} (-D_i - E_j) = - \left( \min_{i+j=k} (D_i + E_j) \right) = -C_k. \]
Therefore, the desired results $C_k$ for the original $(\min,+,M)$-convolution instance are obtained by
computing $C_k = -C'_k$ for every $k \in M$.
Note that the construction of sequences $A$ and $B$ requires $O(n)$ time, and the recovery of $C_k$
from $C'_k$ takes $O(m)$ time. This is hence an efficient, linear-time reduction.

\subsection{Reduction from (max,+,M)-convolution to Positive (max,+,M)-convolution} \label{subsec:maxpmconv}
\paragraph{The positive $(\max,+,M)$-convolution problem.}
Given two sequences $A = (A_0, A_1, \dots, A_{n-1})$ and $B = (B_0, B_1, \dots, B_{n-1})$, such that $A_i
\ge 0$ for all $i$ and $B_j \ge 0$ for all $j$. Additionally, a set of $m$ distinct indices
$M = \{k_0, k_1, \dots, k_{m-1}\}$ is provided, where each $k_s \in \{0, 1, \dots, n-1\}$. The Positive
$(\max,+,M)$-convolution problem is to compute the sequence $C_M =
(C_{k_s})_{s=0}^{m-1}$ where for each $k_s \in M$,
\[ C_{k_s} = \max_{\substack{i+j=k_s}} (A_i + B_j). \]

\paragraph{Reduction:}
Suppose we are given an instance of the $(\max,+,M)$-convolution problem, with sequences $A, B$ (which may
contain negative values) and a set of $m$ indices $M$. Our goal is to compute $C_k =
\max_{i+j=k} (A_i + B_j)$ for each $k \in M$, using an oracle for Positive $(\max,+,M)$-convolution which
requires non-negative input sequences.

First, find the minimum element $\Delta$ present in either sequence $A$ or $B$:
\[ \Delta = \min \left( \left(\min_{0 \le i \le n} A_i\right), \left(\min_{0 \le j \le n} B_j\right) \right). \]
If $\Delta \ge 0$, both sequences $A$ and $B$ are already non-negative. In this scenario, $A, B,$
and $M$ can be passed directly to the Positive $(\max,+,M)$-convolution oracle, and its output will be the
desired $C_M$.

If $\Delta < 0$, we construct two new non-negative sequences $A'$ and $B'$ as follows:
\[ A'_i = A_i - \Delta, \quad \text{and} \quad B'_j = B_j - \Delta, \quad
\text{for } j \in \{0, \dots, n-1\} .\]
Since $A_i \ge \Delta$ and
 $B_j \ge \Delta$ for all respective $i,j$, it is guaranteed that $A'_i \ge 0$ and $B'_j \ge 0$. We
 then call the Positive $(\max,+,M)$-convolution oracle with these non-negative sequences $A', B'$, and
 the set of indices $M$. Let the results from the oracle be $C'_k$ for $k \in M$. For each $k \in
 M$, the oracle computes:
\begin{align*}
    C'_k = \max_{\substack{i+j=k}} (A'_i + B'_j)
         = \left( \max_{\substack{i+j=k}} (A_i + B_j) \right) - 2\Delta = C_k - 2\Delta.
\end{align*}
Thus we can obtain the required sequence $C_k$ as $C_k=C'_k+2\Delta$.
The computation of $\Delta$ takes $O(n)$ time. The construction of $A'$ and $B'$ also takes $O(n)$
time, and recovering the final results $C_k$ from $C'_k$ takes $O(m)$ time. This is hence an
efficient, linear-time reduction.

\subsection{From positive (max,+,M)-convolution to batched MaxRS} \label{subsec:batmaxrs}
\paragraph{Reduction:}
Let an instance of Positive $(\max,+,M)$-convolution be given by two sequences $A,B$
where $A_i \ge 0, B_j \ge 0$ for all $i,j$, and a set
of $m$ target indices $M = \{k_0, k_1, \dots, k_{m-1}\}$, where each $k_s \in \{0, 1, \dots, n-1\}$. We
wish to compute $C_{k_s} = \max_{i+j=k_s} (A_i+B_j)$ for each $k_s \in M$.

In this subsection, we represent a weighted point $p$, located at x-coordinate $x_p$,
and with weight $w_p$, as $p=(x_p,w_p)$. We construct an instance of batched MaxRS as follows:

\begin{enumerate}
    \item \textbf{Point set construction:} Create a set $P_S$ of $4n$ weighted points.
        \begin{itemize}
            \item For each $A_i$ where $i \in \{0, \dots, n-1\}$, create a point $p_i^A = (i, A_i)$,
                and its `guard' point $p'^A_i = (i-0.5, -A_i)$.
            \item Let $X_\text{offset} = 2n-1$. For each $B_j$ where $j \in \{0, \dots, n-1\}$,
                create a point $p_j^B = (X_\text{offset}-j, B_j)$, and its guard point
                $p'^B_j=(X_\text{offset}-j+0.5, -B_j)$.
\begin{figure}[ht]
    \centering
    \includegraphics[width=0.9\textwidth, page=3]{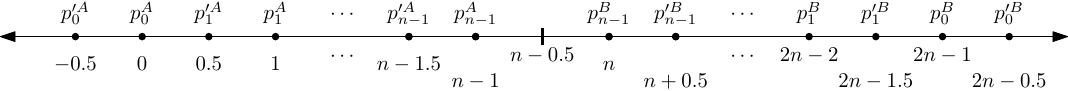}
    \caption{Position of constructed (weighted) points on real line.}
    \label{fig:arrayBRS}
\end{figure}
\end{itemize}

    \item \textbf{Query interval lengths:} For each index $k_s \in M$, we define a query interval
        length as $L_s=2n-1-k_s$ for the batched MaxRS problem.
        Note that since $0\le k_s\le n-1$, the interval length  $L_s\ge n$.
\end{enumerate}
We then query the batched MaxRS oracle with the point set $P$ and the $m$ interval lengths $\{L_s\}_
{k_s \in M}$. Let $W_s^*$ be the maximum sum returned by the oracle for the interval of length
$L_s$. We want to claim that this value $W_s^*$ is in fact the required value $C_{k_s}$.

\begin{restatable}{lemma}{BRSproof}\label{lem:BRSproof}
For each $k_s \in M$, the maximum weight $W_s^*$ found by the batched MaxRS oracle for an interval
of length $L_s = 2n-1-k_s$, is equal to $C_{k_s}^* = \max_{i+j=k_s} (A_i+B_j)$.
\end{restatable}
\begin{proof}[Proof Sketch.]
    Since every point has an
    associated guard point which negates its weight, any placement of an interval can cover the
    weight of at most one $A$-point, and at most one $B$-point. Since placing an interval of length
    $L_s$, starting at the coordinate of some $A$-point $p^A_i$, covers weight $A_i+B_j$ for
    $j=k_s-i$ (proved in Lemma~\ref{clm:interval_wt}), we can discard all other placements of such a segment (since they will either cover
    negative weight, zero weight, or weight of a single element $A_i$ or $B_j$). Using the batched
    MaxRS oracle gives us the maximum possible such sum, which is hence the maximum possible sum of
    $A_i+B_j$ where the indices satisfy $i+j=k_s$.
\end{proof}

We now formalize the above intuition and give a rigorous proof. Let $w(I)$ be the weight covered by some interval $I$.
\begin{claim}\label{clm:interval_wt}
    The weight covered by an interval $I_{i,j}=[i,(2n-1)-j]$ for some $i,j\in\{0,\ldots,n-1\}$, is given by $w(I_{i,j})=A_i+B_j$.
\end{claim}
\begin{proof}
    Since $I_{i,j}$ starts at x-coordinate $i$, it includes point $p_{i}^A=(i, A_{i})$, but not the
    corresponding guard point $p'^{A}_{i}=(i-0.5, -A_{i})$. Similarly, since $I_{i,j}$ ends at
    $(2n-1)-j$, it includes the point $p_{j}^B = ((2n-1)-j, B_{j})$, but not its corresponding guard
    point $p'^{B}_{j} = ((2n-1)-j+0.5, -B_{j})$.

    Further for all $i<i'\le n-1$ and $0\le j'<j$, the interval $I_{i,j}$ covers the points
    $p_{i'}^A$, $p'^A_{i'}$, $p_{j'}^B$, $p'^B_{j'}$. By construction, if a point and its
    corresponding guard point are picked, then their weights sum to zero.
    Hence we get that $w(I_{i,j})=A_{i}+B_{j}$.
\end{proof}

\begin{proof}[Proof of Lemma~\ref{lem:BRSproof}]
Since $A_i, B_j \ge 0$ for all $i,j \in \{0, \dots, n-1\}$,
it follows that $C_{k_s}^* \ge 0$. Let $(i_0, j_0)$ be a pair of indices such
that $i_0, j_0 \in \{0, \dots, n-1\}$, $i_0+j_0=k_s$, and $A_{i_0}+B_{j_0}=C_{k_s}^*$.
Consider an interval $I_{i_0,j_0}=[i_0, (2n-1)-j_0]$. The length of this interval is $x_{end}-x_{start}
=(2n-1)-(i_0+j_0)=(2n-1)-k_s$, which is precisely $L_s$, the query interval
length for the target index $k_s$. Further from Claim~\ref{clm:interval_wt} we know that
$w(I_{i_0,j_0})=A_{i_0}+B_{j_0}=C_{k_s}^*$

Since the batched MaxRS oracle returns $W_s^*$, the maximum possible sum for an interval of length
$L_s$, we have $W_s^* \ge C_{k_s}^*$.
To show that $W_s^* = C_{k_s}^*$, we argue that there exists an optimal interval $I^* = [x^*, x^*+L_s]$
which is of the form $[i, (2n-1)-j]$ for some $i,j\in\{0,\ldots,n-1\}$ such that $i+j=k_s$.

\begin{enumerate}
  \item If $x^*+L_s<n-0.5$ (i.e., \textbf{$I^*$ covers no $B$-points}): then the interval covers
      all points starting from $p'^A_0$ up to some point $p^A_a$ or  $p'^A_a$ for some
      $a\in\{0,\dots,n-1\}$ (this is because $L_s\ge n$). This interval hence either covers zero weight
      (in the former case, when all points covered are paired with their corresponding guard
      points), or negative weight (in the latter case, where only $p'^A_a$ is picked but not $p^A_a$).

 \item If $x^*\ge n-0.5$ (i.e., \textbf{$I^*$ covers no $A$-points}): then the interval
     covers all points starting from some point $p^B_b$ or  $p'^B_b$ for some $b\in\{0,\dots,n-1\}$
     up to the point $p'^B_0$. This interval hence either covers zero weight (in the former case, when
     all covered are paired with their corresponding guard points), or negative weight (in the latter
     case, where only $p'^B_b$ is picked but not $p^B_b$).

 \item We now know that the left endpoint lies to the left of $n-0.5$, and the right end point lies
     to the right of $n-0.5$. If the left endpoint satisfied $x^*<-0.5$ (i.e., \textbf{$I^*$ covers all $A$-points}),
     then all points $p^A_a$ for
     $a\in\{0,\dots,n-1\}$ will be paired with their guard points, and hence contribute a total
     weight of zero. Further, at most one point $p^B_b$ for $b\in\{0,\dots,n-1\}$ can be left
     unpaired, and hence the weight covered by the interval is either zero or $B_b$ for some $b$.
     Similarly, if the right endpoint satisfied $x^*+L_s>2n-0.5$ (i.e., \textbf{$I^*$ covers all $B$-points}),
     then the cost covered by the
     interval can be either zero or $A_a$ for some $a$.

 \item We are now left with the case where \textbf{$I^*$ covers some $A$-points, and some
     $B$-points}. Consider an interval starting at an integral coordinate $i\in\{0,\dots,n-1\}$.
     Given a length $L_s$, the interval would end at some $i+L_s=(2n-1)-(k_s-i)=(2n-1)-j$ (for
     $j\in\{0,\dots,n-1\}$ satisfying $i+j=k_s$). By Claim~\ref{clm:interval_wt} this interval
     covers weight $A_i+B_j$. Now, if we shift the interval left by some $\delta\in(0,0.5]$, the
     interval will no more cover point $p^B_j$, and could potentially additionally cover the point
     $p'^A_i$, i.e., drop a positive weight point, and include a negative weight point. This
     operation can hence not lead to an optimal placement. Similarly, if we move the interval right
     by some $\delta\in(0,0.5]$, the interval will no more cover point  $p^A_i$, and potentially
     cover point  $p'^B_j$ leading again to a sub-optimal placement. Hence, there exists an optimal
     placement of the interval starting (and ending) at integral coordinates.
\end{enumerate}
Thus, an optimal interval $I^*$ must be of the form $I_{i^*,j^*}=[i^*, (2n-1)-j^*]$ for some $i^*,
j^* \in \{0, \dots, n-1\}$.  The length constraint $L_s=(2n-1)-k_s$ implies $i^*+j^*=k_s$. The
batched MaxRS oracle hence returns the cost
$W_s^* = \max_{i^*+j^*=k_s} w(I_{i^*,j^*})= \max_{i^*+j^*=k_s} (A_{i^*}+B_{j^*}) = C_{k_s}^*$.
\end{proof}

This chain of reductions establishes the proof of hardness
of batched MaxRS (Theorem~\ref{thm:batched-maxrs}).

\section{Lower bound for batched smallest k-enclosing interval}
In this section, we study a closely related problem, that is the batched version of the smallest
$k$-enclosing interval (BSEI). Given a set of points on the real line, the smallest $k$-enclosing
interval problem requires us to find the smallest interval that encloses $k$ points.  More
formally, in the batched version of the problem, we are given a set of $n$ points $p_1,
p_2,\ldots,p_n$ on the real line. The goal is to find the smallest
interval enclosing $k$ points, for all $k\in[n]$.


There is a quadratic time algorithm for the problem, and we show that there is also a matching quadratic lower
bound conditioned on the conjectured hardness of the $(\min,+)$-convolution problem.
We achieve the required result by the series of reductions: $(\min,+)$-convolution $\le$ monotone
$(\min,+)$-convolution $\le$ BSEI. The problem definitions and the reductions are given below.

\subsection{Reduction from (min,+)-convolution to monotone (min,+)-convolution}
\begin{definition}[Monotone $(\min,+)$-convolution]
    Given two sequences \( D = (D_0, D_1, \dots, D_{n-1}) \) and \( E = (E_0, E_1, \dots, E_{n-1}) \),
such that \( D_0>D_1>\ldots> D_{n-1}\) and \( E_0> E_1> \ldots> E_{n-1}\), the monotone
$(\min,+)$-convolution produces a new sequence \( F \), where for each $k\in\{0,\ldots,n-1\}$,
\[ F_k = \min_{i + j = k} (D_i + E_j). \]
\end{definition}

\paragraph{Reduction:}
We are given unsorted sequences $A,B$, as part of the instance of $(\min,+)$-convolution. We will generate two
monotone sequences $D,E$, to create an instance of the monotone $(\min,+)$-convolution problem, such
that solving the monotone instance would give us a solution for the original instance.
\begin{enumerate}
    \item Let $\Delta=1+\max_{i\in\{1,\ldots,n-1\}} \max((A_i-A_{i-1}),(B_i-B_{i-1}))$. So the new sequences are, for all
$i\in\{0,\ldots,n-1\}$, $D_i=A_i-i\cdot\Delta$ and  $E_i=B_i-i\cdot\Delta$.
Note that for all $i\in[n-1], D_{i-1}-D_{i} = A_{i-1}-A_{i} + \Delta > 0$, and hence the sequence is
strictly decreasing (and similarly sequence $E$
is strictly decreasing).

    \item Solve the monotone $(\min,+)$-convolution problem to obtain the sequence $F$ such that
\[ F_k=\min_{i+j=k} (D_i+E_j) =\min_{i+j=k}((A_i-i\cdot\Delta)+(B_j-j\cdot\Delta))
=C_k - k\cdot\Delta. \]
    \item Obtain the solution to the $(\min,+)$-convolution problem as $C_k=F_k+k\cdot\Delta$.
Note that the reduction only takes linear time, to compute $\Delta$, and then subtract it from every
element of the input sequences $A$ and $B$.
\end{enumerate}

\subsection{Reduction from monotone (min,+)-convolution to BSEI}
\paragraph{Reduction:}
We are given monotone sequences $D,E$ as part of an instance of monotone  $(\min,+)$-convolution.
Consider the sequence $P$ of $2n$ elements defined as follows: For each
$i\in\{0,\ldots,n-1\}, P_i=-D_{i}+(D_{n-1}-1)$, and $P_{n+i}=E_{(n-1)-i}+(1-E_{n-1})$.
Notice that the added offsets of $(D_{n-1}-1)$ ($(1-E_{n-1})$ resp.) ensure that
$P_i$ ($P_{n+i}$ resp.) remains negative (positive resp.).

Now let us construct an instance of BSEI with $2n$ points, and the location of the $i$-th
point being given by $P_i$.
See Figure~\ref{fig:arrayred} to see the locations of points in this instance.
Now, if we could solve the BSEI problem, we would obtain a sequence
$G_1,G_2,\ldots,G_{2n}$, where $G_{k}$ denotes the minimum length of an interval covering $k$
points.

\begin{figure}[ht]
    \centering
    \includegraphics[width=0.8\textwidth]{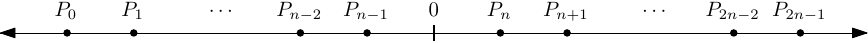}
    \caption{Values of newly constructed sequence as points on the real line.}
    \label{fig:arrayred}
\end{figure}

The distance between arbitrary (say $i$-th) point to its $\ell$-th successor is
$P_{\ell+i}-P_i$, and this covers $\ell+1$ points. So for $k\in\{0,\dots.n-1\}$, we can
write the term $F_{k}$ as,
\begin{align*}
  F_k &= \min_{i\in\{0,\ldots,k\}} \left( E_{k-i} + D_i \right)
        = \min_{i\in\{0,\ldots,k\}} \left( E_{(n-1)-(n-1-k+i)} + D_i \right) \\
        &= (D_{n-1}-1)-(1-E_{n-1})+\min_{i\in\{0,\ldots,k\}} \left( P_{n+(n-1-k+i)} - P_i \right) \\
        &= G_{2n-k} + D_{n-1} + E_{n-1} -2.  
\end{align*}


So, if we could solve BSEI, we would be able to get the solution to the instance of monotone $
(min,+)$-convolution as $F_k=G_{2n-k}+ D_{n-1} + E_{n-1} -2$ for $k\in\{0,1,\ldots,n-1\}$. Note that this reduction only
takes linear time, as we only have to construct the new array by iterating over the given sequences
once. This finishes the proof of Theorem~\ref{thm:bsei}.

\section{Future work}
We conclude the paper with a few open problems.
\begin{itemize}[itemsep=1ex]
    \item Our technique based on output sensitivity and sampling colors in Section~\ref{sec:outsens-colorsample} to obtain $(1-\eps)$-approximation for colored disk MaxRS in $\bbR^2$ (Theorem~\ref{thm:col-balls-1minus-eps-approx}) is quite general, and we believe it will be useful for solving other variants of MaxRS. An immediate open problem is to extend this technique to answer colored MaxRS for boxes and balls in $\bbR^3$.
    
    \item For the batched MaxRS problem in $\bbR^1$, the conditional hardness result\textemdash lower bound of $\Omega(mn)$\textemdash holds for the weighted set of points (Theorem~\ref{thm:batched-maxrs}). It remains open whether a similar lower bound can be established for the unweighted case in $\bbR^1$. The best known upper bound for batched MaxRS for rectangles in $\bbR^2$ is still $O(mn\log n)$ in the unweighted setting (by running the exact $O(n\log{n})$ time algorithm~\cite{imai1983finding, nandy1995unified} for each of the $m$ ranges).

    \item The batched MaxRS problem for  disks in $\bbR^2$ can be solved in $O(mn^2)$ time (by running the exact $O(n^2)$ time algorithm~\cite{chazelle1986circle} for each of the $m$ ranges). Can we prove a matching lower bound?
    
    \item Finally, for the static MaxRS for $d$-balls, a simple upper bound of $\Ot(n^d)$ can be obtained by computing the entire arrangement. It would be interesting to prove if there is a matching lower bound. Currently, such a conditional hardness result is known for disks~\cite{aronov2008approximating}.
\end{itemize}

\bibliographystyle{alpha}
\bibliography{references}
\end{document}